\documentclass[conference,letterpaper]{IEEEtran}

\usepackage[utf8]{inputenc} 
\usepackage[T1]{fontenc}
\usepackage{url}
\usepackage{ifthen}
\usepackage{cite}
\usepackage[cmex10]{amsmath} % Use the [cmex10] option to ensure complicance
\usepackage{amsthm}
\usepackage{booktabs}

\usepackage{tikz}
\usepackage{tikzscale}
\usepackage{amssymb}
\usepackage{mathtools}

\usepackage[colorlinks=true,linkcolor=black,anchorcolor=black,citecolor=black,filecolor=black,menucolor=black,runcolor=black,urlcolor=black]{hyperref}

\usepackage{pdfcomment}%remove after edits are done

\newtheorem{theorem}{Theorem}

\newtheorem{lemma}{Lemma}

\newtheorem{prop}{Proposition}

\theoremstyle{definition}

\newtheorem{definition}{Definition}
\newtheorem{remark}{Remark}
\newtheorem{algorithm}{Algorithm}

\newcommand{\off}[1]{}

\newcommand{\p}{\mathrm{Pr}}

\begin{document}
\title{A Monotone Circuit Construction for\\ Individually-Secure Multi-Secret Sharing} 

%%% Many authors with many affiliations:
 \author{
   \IEEEauthorblockN{Cailyn Bass\IEEEauthorrefmark{1},
                     Alejandro Cohen\IEEEauthorrefmark{2},
                     Rafael G. L. D'Oliveira\IEEEauthorrefmark{1},
                     and Muriel M\'edard\IEEEauthorrefmark{3}}
   \IEEEauthorblockA{\IEEEauthorrefmark{1}
                    School of Mathematical and Statistical Sciences, Clemson University, Clemson, SC, USA,
                     \{cailynb,rdolive\}@clemson.edu}
   \IEEEauthorblockA{\IEEEauthorrefmark{2}
                    Faculty of Electrical and Computer Engineering, Technion--Institute of Technology, Haifa, Israel,
                     alecohen@technion.ac.il}
   \IEEEauthorblockA{\IEEEauthorrefmark{3}
                     Research Laboratory of Electronics, Massachusetts Institute of Technology, Cambridge, MA, USA,
                     medard@mit.edu}
 }

\maketitle

%%%%%%
%% Abstract: 
%% If your paper is eligible for the student paper award, please add
%% the comment "THIS PAPER IS ELIGIBLE FOR THE STUDENT PAPER
%% AWARD." as a first line in the abstract. 
%% For the final version of the accepted paper, please do not forget
%% to remove this comment!
%%

\begin{abstract}
   In this work, we introduce a new technique for taking a single-secret sharing scheme with a general access structure and transforming it into an individually secure multi-secret sharing scheme where every secret has the same general access structure. To increase the information rate, we consider \emph{Individual Security} which guarantees zero mutual information with \emph{each secret individually}, for any unauthorized subsets. Our approach involves identifying which shares of the single-secret sharing scheme can be replaced by linear combinations of messages. When $m-1$ shares are replaced, our scheme obtains an information rate of $m/|S|$, where $S$ is the set of shares. This provides an improvement over the information rate of $1/|S|$ in the original single-secret sharing scheme.
\end{abstract}

\section{Introduction}

A secret sharing scheme is a method for sharing a secret amongst a set of participants in such a way that only certain authorized subsets of the participants are able to retrieve the secret. Any unauthorized subset that combines their shares should gain no new information about the secret. The first secret sharing schemes were introduced in 1979 by Blakely \cite{blakely} and Shamir \cite{shamir}. These initial schemes are now called $(k,n)$-threshold schemes since any set of participants of size greater than or equal to $k$ is authorized to compute the secret. In 1988 Ito et. al. \cite{ito} introduced a multiple assignment scheme to produce secret sharing schemes for any general access structure. Monotone Boolean functions were then utilized by Benaloh et. al. \cite{benaloh} in 1992 to improve the efficiency of the multiple assignment scheme. These functions are the main idea behind the monotone circuit construction \cite{stinson} that we use throughout this paper.  

The notion of security utilized in the previous secret sharing works is that of perfect security, first introduced by Shannon \cite{shannon1949communication} in 1949.\footnote{Shannon published an earlier version of this paper in 1945 \cite{shannon1945mathematical} which was classified. Interestingly, this precedes Shannon's \emph{other} seminal paper \cite{shannon1948mathematical}.} This same notion was utilized by Wyner in 1975 when proposing the wire-tap channel \cite{wyner1975wire,ozarow1984wire}. Carleial and Hellman then proposed in 1977 the notion of individual security \cite{Hellman} to obtain higher data rates for the case where messages are uniformly and independently distributed. Since then, individual security has been applied to various communication and storage systems, e.g., single communication link \cite{SCMUniform}, broadcast channels \cite{mansour2014secrecy,chen2015individual,mansour2015individual,mansour2016individual}, multiple-access channels \cite{goldenbaum2015multiple,chen2016secure}, networks and multicast communications \cite{bhattad2005weakly,silva2009universal,cohen2018secure}, algebraic security \cite{lima2007random,claridge2017probability}, terahertz wireless systems \cite{cohen2023absolute}, angularly dispersive links \cite{yeh2023securing}, and distributed storage systems \cite{kadhe2014weakly,kadhe2014weakly1,paunkoska2016improved,paunkoska2018improving,bian2019optimal}. Individual security guarantees that an eavesdropper, which can obtain any subset (up to a certain size) of the shared information, obtains no information about each message individually. Yet, an eavesdropper may obtain some insignificant controlled information about the mixture of all the messages.

Karnin et. al. in 1983 \cite{hellman2} brought the notion of individual security to secret sharing for the case of multiple secrets. This setting is known as multi-secret sharing.\footnote{In order to better distinguish the settings, in the rest of the paper we refer to the traditional secret sharing as single-secret sharing.} Since then, much work has been done in multi-secret sharing, e.g., works on threshold schemes \cite{hellman2,multiThreshold,thresholdConstruction,multisecret}, generalizing ideal secret sharing schemes to ideal multi-secret schemes \cite{ideal}, and defining different access structures for different messages \cite{crescenzo,multisecret,efficientSharing}. However, to the best of our knowledge, there is no general construction of individually-secure multi-secret sharing schemes for general access structures in the literature. For example, as the monotone circuit construction \cite{stinson} for single-secret sharing.

Our main contributions are as follows. In Algorithm~\ref{alg: rep} we show how to convert a single-secret sharing scheme with a general access structure into an individually secure multi-secret sharing scheme where each message has the same access structure. The algorithm works by identifying which shares of the single-secret sharing scheme can be replaced by combinations of messages. In Theorem~\ref{teo: construction}, we show that the scheme obtained from Algorithm~\ref{alg: rep} by performing $m-1$ shares replacements is a multi-secret sharing scheme with information rate of $m/|S|$, where $S$ is the set of shares. This provides an improvement over the information rate of $1/|S|$ in the original single-secret sharing scheme. Moreover, we show in Theorem~\ref{teo: nonRep} that making extra replacements causes the decodability of the scheme to fail.

\begin{table}[!t]
\centering
\begin{tabular}{c c}
\toprule
%\small
Symbol & Description  \\
\midrule
$\mathbb{F}_q$ & Finite Field of size $q$\\
$\mathcal{P}=\{P_1,\ldots,P_n\}$ & Set of participants \\
$n$ & Number of participants\\
$\Gamma$ & Access structure\\ & (set of all authorized subsets)\\
$A_i\subseteq\mathcal{P}$ & A minimal authorized subset\\
$\Gamma_0=\{A_1,\ldots,A_r\}$ & Basis of the access structure\\ & (set of minimal authorized subsets)\\
$r$ & Number of minimal authorized subsets\\
$2^\mathcal{P}\setminus\Gamma$ & Set of unauthorized subsets\\
$U$ & An unauthorized subset\\
$S_A$ & Set of shares belonging to\\ & authorized subset $A$\\
$S_j^A$ & Share of $P_j$ associated with $A$\\
$S_{P_j}=\{S_j^A\}_{A\in\Gamma_0}$ & Set of shares belonging to $P_j$\\
$S=\bigcup_{P_j\in\mathcal{P}}{S_{P_j}}$ & Set of all shares\\
$S_U$ & Set of shares belonging to\\ & an unauthorized subset\\
$m$ & Number of messages\\
$M=\{M_1,\ldots,M_m\}$ & Set of all messages\\
$\mathcal{R}_{SS}$ & Information rate of a\\ & secret sharing scheme\\
$\mathcal{R}_{MS}$ & Information rate of a\\ & multi-secret sharing scheme\\
\bottomrule \vspace{-5pt}
\end{tabular}
\caption{List of Symbols}
\end{table}

\section{Preliminaries on Single-Secret Sharing}

In this section, we first give definitions for a single-secret sharing scheme on a general access structure and its information rate. We then give an overview of the monotone circuit construction from \cite{stinson}.

\subsection{Single-Secret Sharing}

Given a group of $n$ participants $\mathcal{P}=\{P_1,\ldots,P_n\}$ a secret message $M_1$, and a set of $k$ authorized subsets $\Gamma = \{A_1, \ldots, A_k\}$, which we call an access structure, a single-secret sharing scheme consists in assigning to each participant a set of shares $S_{P_1}, \ldots, S_{P_n}$ in such a way that only authorized subsets of $\mathcal{P}$ are able to retrieve the secret, while unauthorized subsets $U \notin \Gamma$ remain completely ignorant about the secret.

As is common in the literature \cite{ito,benaloh,stinson}, we only consider monotone access structures, i.e., if $A\in \Gamma$ and $A'\supseteq A$, then $A'\in\Gamma$. Thus, it is sufficient to consider only the minimal authorized subsets $A_i\subseteq\mathcal{P}$. The set of minimal authorized subsets is called a basis for $\Gamma$ which we denote by $\Gamma_0$.

We denote the set of all shares by $S=\bigcup_{P_j\in\mathcal{P}}{S_{P_j}}$, the shares held by an authorized subset by $S_A\subseteq S$, and the shares held by an unauthorized subset by $S_U\subseteq S$. This leads to the following definition.

\begin{definition}
    Given a basis $\Gamma_0=\{A_1,\ldots,A_r\}$ for an access structure $\Gamma$ and a secret $M_1\in\mathbb{F}_q$, a \emph{single-secret sharing scheme realizing $\Gamma$} is one in which a set of shares $S\subseteq\mathbb{F}_q$ is created such that the following hold.
    \begin{enumerate}
        \item \textbf{Decodability}: The conditional entropy $\mathrm{H}(M_1|S_A)=0$ for all $A\in\Gamma$. In other words, every authorized subset is able to compute the secret. 
        
        \item \textbf{Security}:  The mutual information $\mathrm{I}(M_1;S_U)=0$ for all $U\in 2^{\mathcal{P}}\setminus \Gamma$. In other words, no unauthorized subset learns any new information about the secret.
    \end{enumerate}
\end{definition}

\begin{remark}
    We only need to consider access structures where every authorized subset consists of more than one participant. If $\Gamma$ is an access structure on $n$ participants that has an authorized subset $A=\{P_i\}$ for some $i$, then we can consider the access structure $\Gamma'$ on $n-1$ participants where $\Gamma'_0=\Gamma_0\setminus A$. This is due to the fact that $P_i$'s share would be the secret itself.
\end{remark}

We now define our performance metric.

\begin{definition}
    The \emph{information rate} of any secret sharing scheme (single or multi) is measured in terms of the ratio between the number of secrets and the total number of shares.
\end{definition}

\subsection{Monotone Circuit Construction}

The monotone circuit construction \cite{stinson} is a single-secret sharing scheme that can realize any access structure $\Gamma$. The scheme works as follows. Given a set of participants $\mathcal{P}$ and a basis of authorized subsets $\Gamma_0=\{A_1,\ldots,A_r\}$ of $\Gamma$, we associate a monotone Boolean function that represents $\Gamma_0$. For example, if $\Gamma_0=\{A_1,A_2\}$ where $A_1=\{P_1,P_2\}$ and $A_2=\{P_1,P_3\}$, then the monotone function representing $\Gamma_0$ is $(P_1 \land P_2)\lor(P_1 \land P_3)$. This monotone Boolean function is then utilized to create a circuit where the gates of the circuit correspond to the clauses of the function. Share assignment for the monotone circuit construction is carried out as follows. For participant $P_j\in A$, the share assigned to $P_j$ associated with $A$ is denoted by $S_j^{A}$. These shares are chosen uniformly at random but in such a way that their sum is equal to $M_1$, i.e., $\sum_{P_j\in A}{S_j^{A}}=M_1$. We achieve this by fixing one of the shares and choosing the others uniformly at random.

We remark here that to simplify the technical aspects and allow us to focus on the key methods and results, during this paper we have chosen to work with the monotone circuit construction for its simplicity. However, it is known that there are other, more efficient, constructions for general access structure secret sharing such as the vector space construction of \cite{brickell1989some}, the decomposition construction of \cite{stinson}, and the geometric construction of \cite{simmons1988really,simmons1991introduction,77a41a127060489e8ca2a6e9822932d5}.

\section{Multi-Secret Sharing}

Given a group of participants $\mathcal{P}$ and an access structure $\Gamma$ as before along with a set of uniform and independently distributed messages $M=\{M_1,\ldots,M_m\}$, a multi-secret sharing scheme consists in creating a set of shares $S$ in such a way that authorized subsets are able to reconstruct all $m$ messages while unauthorized subsets remain completely ignorant about each individual message. Schemes that have this property are said to be individually secure. The following is a more formal definition.

\begin{definition}
    Given a basis $\Gamma_0=\{A_1,\ldots,A_r\}$ for an access structure $\Gamma$ and a set of messages $M=\{M_1,M_2,\ldots,M_m\}\subseteq\mathbb{F}_q$ where each message is uniformly and independently distributed, a \emph{multi-secret sharing scheme realizing $\Gamma$} is one in which a set of shares $S\subseteq\mathbb{F}_q$ is distributed amongst the participants such that the following hold.
    \begin{enumerate}
        \item \textbf{Decodability}: The conditional entropy $\mathrm{H}(M|S_A)=0$ for all $A\in\Gamma$. In other words, every authorized subset can compute all $m$ messages. 

        \item \textbf{Individual Security}: For any unauthorized subset $U\in 2^{\mathcal{P}}\setminus \Gamma$, the mutual information $\mathrm{I}(M_\ell;S_U)=0$, for any $\ell\in\{1,\ldots,m\}$. In other words, unauthorized subsets learn no new information about each message individually.
    \end{enumerate}
\end{definition}

In Theorem~\ref{teo: construction} we show how to transform a single-secret sharing scheme for a general access structure into a multi-secret sharing scheme for the same access structure. In the next section, we present a detailed example of our approach to showcase the essential ingredients of the scheme.

\section{An Example of Our Approach} \label{sec: example}

In the example given herein, we show how to convert a single-secret sharing scheme from \cite{stinson}, with a general access structure, into an individually secure multi-secret sharing scheme for the same access structure. Our technique consists in identifying shares from the original scheme which can be substituted with a linear combination of messages. We begin by presenting a single-secret sharing scheme.

\subsection{Single-Secret Sharing for a General Access Structure}

Let $\mathcal{P}=\{P_1,P_2,P_3,P_4\}$ be the set of participants and consider an access structure $\Gamma$ with a basis \[\Gamma_0=\{\{P_1,P_2,P_4\},\{P_1,P_3,P_4\},\{P_2,P_3\}\}.\] Thus, for example, participants $P_1,P_2,\text{ and } P_4$ can compute the secret but participants $P_1$ and $P_2$ cannot. Also, since participants $P_2$ and $P_3$ can compute the secret, so can participants $P_1$, $P_2$, and~$P_3$, i.e., $\{P_1, P_2, P_3\} \in \Gamma.$ The access structure $\Gamma_0$ can be represented by the monotone circuit in Figure \ref{fig: 1}.

    \begin{figure}[t]
        \centering

        \includegraphics[width = \columnwidth]{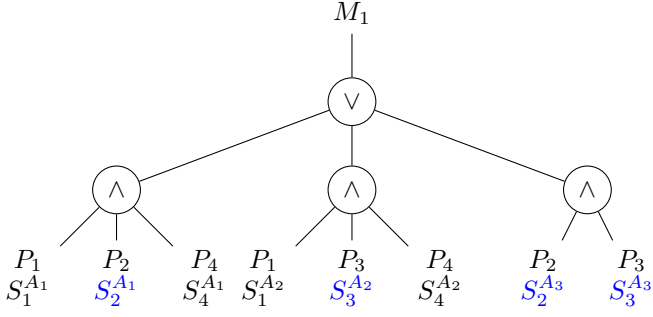}

        \caption{The monotone circuit for $\Gamma_0$, which can be represented by the following monotone Boolean function: $(P_1\land P_2\land P_4)\lor(P_1\land P_3\land P_4)\lor(P_2\land P_3)$. Each $\land$ gate corresponds to adding the shares on the input wires and thus represents a secret sharing scheme among the participants of $A_i$. These smaller schemes are the key to proving the monotone circuit single-secret sharing scheme is decodable and secure. The replaceable shares for our construction are shown in blue.} \label{fig: 1}
    \end{figure}

The construction in \cite{stinson} consists in assigning one share to each participant corresponding to each authorized subset $A \in \Gamma_0$. Denote the authorized subsets in $\Gamma_0$ by $A_1 = \{ P_1, P_2, P_4 \}$, $A_2 = \{ P_1, P_3, P_4 \}$, and $A_3 = \{ P_2, P_3\}$. Then by \cite{stinson}, participant $P_1$ is assigned two shares, $S_1^{A_1}$ and $S_1^{A_2}$, participant $P_2$ is assigned two shares, $S_2^{A_1}$ and $S_2^{A_3}$, participant $P_3$ is assigned two shares, $S_3^{A_2}$ and $S_3^{A_3}$, and participant $P_4$ is assigned two shares, $S_4^{A_1}$ and $S_4^{A_2}$.

The shares in each authorized subset are chosen uniformly at random but such that their sum is equal to the secret $M_1$. In this example, this means that shares are constructed so that $S_1^{A_1}+S_2^{A_1}+S_4^{A_1}=M_1$, $S_1^{A_2}+S_3^{A_2}+S_4^{A_2}=M_1$, and $S_2^{A_3}+S_3^{A_3}=M_1$. One way to obtain this is by fixing one of the $S_j^{A_1}$'s, one of the $S_j^{A_2}$'s, and one of the $S_j^{A_3}$'s and choosing the other shares uniformly at random. For the single-secret sharing problem in \cite{stinson}, the choice of which share to fix is not relevant. In other words, fixing $S_2^{A_1}$ instead of fixing $S_4^{A_1}$ as is done in \cite{stinson} does not affect the decodability or security of the scheme. This choice, however, is relevant in our construction when we replace shares with linear combinations of messages.

The decodability of the scheme is shown in \cite{stinson}, and follows from the fact that the shares associated with each authorized subset form an independent single-secret sharing scheme. The security of the scheme also follows from this fact since the schemes being independent means that no share is assigned to more than one participant. Note that since there are $8$ total shares and $1$ secret, the information rate of this single-secret sharing scheme is $\mathcal{R}_{SS}=\frac{1}{8}$. 

\subsection{The Conditions for Replacing a Share}

In order to obtain a multi-secret sharing scheme with the same access structure, $\Gamma$, for all messages, our construction consists of replacing certain random shares from the single-secret sharing scheme with linear combinations of messages. The criterion for identifying which shares can be replaced is as follows. A share $S_j^{A_i}$ is replaceable if for every $A'\in\Gamma_0$ either $P_j \in A'$ or $A\setminus\{P_j\}\subseteq A'$. Thus, in our example, the share $S_2^{A_1}$ is replaceable since $P_2 \in A_1 \cap A_3$ and $A_1 \setminus P_2 \subseteq A_2$. However, the share $S_1^{A_1}$ is not, since $P_1 \notin A_3$ and $A_1 \setminus P_1 \not\subseteq A_3$.

The first step in our construction consists in determining which shares are replaceable. Checking for the replaceability conditions, we obtain that the replaceable shares for the access structure $\Gamma_0$ are $S_2^{A_1},S_3^{A_2},S_2^{A_3},$ and $S_3^{A_3}$, as illustrated in Figure \ref{fig: 1}. To show why the replaceability conditions are defined as so, consider the following two examples where we want to introduce a new message $M_2$ into the shares. For these examples, we choose $S_4^{A_1}=M_1-S_1^{A_1}-S_2^{A_1}$ to be the fixed share of $A_1$.

Consider replacing the share $S_2^{A_1}$ with the linear combination of messages $2M_1+M_2$. The authorized set $A_1$ can compute $M_1$ because the shares in $A_1$ still sum to $M_1$. They can then compute $M_2$ since they can subtract $2M_1$ from $S_2^{A_1} = 2M_1+M_2$. Again, the authorized set $A_2$, can compute $M_1$ because its shares still sum to $M_1$. They can then compute $M_2$ since $M_2 = -S_1^{A_1} - S_4^{A_1} - M_1$. Finally, the authorized set $A_3$ can compute $M_1$ because the shares in $A_3$ still sum to $M_1$. Since $A_3$ contains participant $P_2$,  the share $S_2^{A_1}=2M_1+M_2\in S_{A_3}$. So they can compute $M_2$ by subtracting $2M_1$ from $S_2^{A_1}$. Thus, after computing $M_1$, all three authorized subsets can compute $M_2$.

We now show that if the replaceability condition is not satisfied, then our replacement technique does not work. As described above, the share $S_1^{A_1}$ does not satisfy the replacement condition. Suppose we instead replace the share $S_1^{A_1}$ with $2M_1+M_2$. We show that the authorized set $A_3$ is not able to compute $M_2$. Note that the only shares $A_3$ has access to are $S_2^{A_1},S_2^{A_3},S_3^{A_2},$ and $S_3^{A_3}$. Since $S_1^{A_1}=M_1-S_2^{A_1}-S_4^{A_1}$ and neither $S_1^{A_1}$ nor $S_4^{A_1}$ are accessible to the participants in $A_3$, there is no linear combination of the shares in $S_{A_3}$ that allow for the recovery of $S_1^{A_1}$. Thus, $M_2$ cannot be computed by the participants in authorized subset $A_3$. 

The second step in our construction is to choose, for each authorized subset $A \in \Gamma_0$, which share $S_j^{A}$ to fix. The choice is made as follows. For each authorized subset $A\in\Gamma_0$, if every share $S_j^{A}$ is replaceable, then pick any of them to be fixed. Otherwise, pick a non-replaceable share $S_j^{A}$ to be fixed. The choice of which share to fix is relevant for our construction since we want to maximize the number of messages we are able to introduce into the scheme. 

The third and final step in our construction is to take the shares that are replaceable, but not fixed, and replace them with linear combinations of the message $M_1$ with a new message $M_\ell$, for each replaced share. These linear combinations are of the form $aM_1+bM_\ell$ where $a \notin \{0,1\}$ and $b\neq0$. It is necessary for both $a$ and $b$ to be different than zero, because no participant should have a share which consists of a single message $M_\ell$, otherwise, we would have an authorized subset with a single participant.   Furthermore, $a\neq1$ since a fixed share $S_j^{A}$ has the form $M_1-\sum_{P_k\in A\setminus\{P_j\}}{S_k^{A}}$. Suppose that we replace the share $S_i^A$ with $aM_1+bM_\ell$ where $a=1$. Then since $A\in\Gamma_0$ is a minimal authorized subset, $A\setminus \{P_i\}$ is an unauthorized subset. If the participants in $A\setminus P_i$ add their shares, they can compute $M_1-S_i^{A}=M_1-M_1-bM_\ell=-bM_\ell$ which violates the security of the scheme.

\subsection{From Single-Secret to Multi-Secret Sharing Schemes}

We now show the details of converting the monotone circuit construction to our construction for the basis $\Gamma_0$ depicted in Figure \ref{fig: 1} using the following steps:
\begin{enumerate}
    \item Identify which shares are replaceable. A share $S_j^{A_i}$ is replaceable if $P_j\in A$ or $A_i\setminus \{P_j\}\subseteq A$ for all $A\in\Gamma_0$.
    \item For each $A\in\Gamma_0$, if the share $S_j^A$ is replaceable for all participants $P_j\in A$, pick one of them to fix. Otherwise, fix a non-replaceable share. 
    \item Replace all shares $S_j^A$ such that $S_j^A$ is replaceable but not fixed. The replacement is of the form $aM_1+bM_\ell$ where $a\notin{0,1}$ and $b\neq 0$. Note that each $M_\ell$ appears in exactly one replacement.
\end{enumerate}

We note that any remaining shares that have not been replaced or fixed are still uniform random over $\mathbb{F}_q$.

As mentioned earlier (and shown in Figure \ref{fig: 1}), the replaceable shares for $\Gamma_0=\{\{P_1,P_2,P_4\},\{P_1,P_3,P_4\},\{P_2,P_3\}\}$ are $S_2^{A_1},S_3^{A_2},S_2^{A_3},$ and $S_3^{A_3}$. We choose to fix the shares $S_4^{A_1}$, $S_4^{A_2}$, and $S_3^{A_3}$. Thus, $S_4^{A_1}=M_1-S_1^{A_1}-S_2^{A_1}$, $S_4^{A_2}=M_1-S_1^{A_2}-S_3^{A_2}$, and $S_3^{A_3}=M_1-S_2^{A_3}$. The replacements we make for this access structure are $S_2^{A_1}=2M_1+M_2$, ~$S_2^{A_3}=2M_1+M_3$, and $S_3^{A_2}=2M_1+M_4$.

We now prove that the scheme described above is decodable (all authorized subsets can compute all messages) and individually secure (no unauthorized subset gains any new information about any individual messages).

\begin{prop}[Decodability]
    $\mathrm{H}(M_1,M_2,M_3,M_4|S_A)=0$ where $S_A$ is the set of shares held by an authorized subset, i.e. every subset $A\in\Gamma_0$ can compute all messages $M_1,M_2,M_3,M_4$.
\end{prop}

\begin{proof}
    $\mathbf{\{P_1,P_2,P_4\}}$: For this authorized subset, 
    $S_A=\{S_1^{A_1},S_1^{A_2},S_2^{A_1},S_2^{A_3},M_1-S_1^{A_1}-S_2^{A_1},
        M_1-S_1^{A_2}-S_3^{A_2}\}.$
    Since $M_1-S_1^{A_1}-S_2^{A_1}+S_1^{A_1}+S_2^{A_1}=M_1$, it follows that $S_2^{A_3}-2M_1=M_2$, $S_2^{A_1}-2M_1=M_3$, and $-(M_1-S_1^{A_2}-S_3^{A_2})-S_1^{A_2}-M_1=M_4.$
    Thus, by Lemma \ref{lem: 2.1}, $\mathrm{H}(M_1,M_2,M_3,M_4|S_A)=0$.

    $\mathbf{\{P_1,P_3,P_4\}}$: For this authorized subset, 
    $S_A=\{S_1^{A_1},S_1^{A_2},S_3^{A_2},M_1-S_2^{A_3},M_1-S_1^{A_1}-S_2^{A_1},
        M_1-S_1^{A_2}-S_3^{A_2}\}.$
    Since $M_1-S_1^{A_2}-S_3^{A_2}+S_1^{A_2}+S_3^{A_2}=M_1$, it follows that $-(M_1-S_1^{A_1}-S_2^{A_1})-M_1-S_1^{A_1}=M_2$, $-(M_1-S_2^{A_3})-M_1=M_3$, and $S_3^{A_2}-2M_1=M_4$. Thus, by Lemma \ref{lem: 2.1} $\mathrm{H}(M_1,M_2,M_3,M_4|S_A)=0$.

    $\mathbf{\{P_1,P_2\}}$: For this authorized subset, $S_A=\{S_2^{A_1},S_2^{A_3},S_3^{A_2},M_1-S_2^{A_3}\}.$ Since $M_1-S_2^{A_3}+S_2^{A_3}=M_1$ it follows that $S_2^{A_1}-2M_1=M_2$, $S_2^{A_3}-2M_1=M_3$, and $S_3^{A_2}-2M_1=M_4$. Thus by Lemma \ref{lem: 2.1}, $\mathrm{H}(M_1,M_2,M_3,M_4|S_A)=0$.
\end{proof}

\begin{prop}[Individual Security]
    $\mathrm{I}(M_\ell;S_U)=0$ for $\ell\in\{1,\ldots,4\}$ where $S_U$ is the set of shares held by unauthorized subset, i.e. no unauthorized subset learns any new information about any of the 4 messages.
\end{prop}

\begin{proof}
    It is sufficient to consider only the maximal unauthorized subsets which are, for this example, $\{P_1,P_2\}$,$\{P_1,P_3\}$,$\{P_1,P_4\}$,$\{P_2,P_4\}$, and $\{P_3,P_4\}$. We prove the proposition for $\{P_1,P_2\}$ since analogous arguments work for the other unauthorized subsets.
    For this subset, $S_U=\{S_1^{A_1},S_1^{A_2},S_2^{A_1},S_2^{A_3}\}$. By the definition of mutual information,
    $\mathrm{I}(M_1;S_U)=\mathrm{H}(S_U)-\mathrm{H}(S_U|M_1).$
    
    Note that $S_U$ can be represented by the following matrix equation:
    \begin{equation*}
        \mathcal{A}
        \left(\begin{IEEEeqnarraybox*}[][c]{,c,}
            M_1\\
            M_2\\
            M_3\\
            M_4\\
            S_1^{A_1}\\
            S_1^{A_2}
        \end{IEEEeqnarraybox*}\right)=
        \left(\begin{IEEEeqnarraybox*}[][c]{,c/c/c/c/c/c,}
            0&0&0&0&1&0\\
            0&0&0&0&0&1\\
            2&1&0&0&0&0\\
            2&0&1&0&0&0
        \end{IEEEeqnarraybox*}\right)
        \left(\begin{IEEEeqnarraybox*}[][c]{,c,}
            M_1\\
            M_2\\
            M_3\\
            M_4\\
            S_1^{A_1}\\
            S_1^{A_2}
        \end{IEEEeqnarraybox*}\right)=
        \left(\begin{IEEEeqnarraybox*}[][c]{,c,}
            S_1^{A_1}\\
            S_1^{A_2}\\
            2M_1+M_2\\
            2M_1+M_3
        \end{IEEEeqnarraybox*}\right)
    \end{equation*}
    Since $\mathcal{A}$ is full rank and $S_1^{A_1},S_1^{A_2},M_1,M_2,M_3,M_4$ all follow a uniform distribution, Lemmas \ref{lem: 2.4} and \ref{lem: 2.3} give that $\mathrm{H}(S_U)=4\log_2{q}$.
    
    Now, by Lemma \ref{lem: 2.2}, $\mathrm{H}(S_U|M_1)=\mathrm{H}(S_1^{A_1},S_1^{A_2},M_2,M_3)$. Thus, $S_U|M_1$ can be represented by the following:
    \begin{equation*}
        \mathcal{A}_1
        \left(\begin{IEEEeqnarraybox*}[][c]{,c,}
            M_1\\
            M_2\\
            M_3\\
            M_4\\
            S_1^{A_1}\\
            S_1^{A_2}
        \end{IEEEeqnarraybox*}\right)=
        \left(\begin{IEEEeqnarraybox*}[][c]{,c/c/c/c/c/c,}
            0&0&0&0&1&0\\
            0&0&0&0&0&1\\
            0&1&0&0&0&0\\
            0&0&1&0&0&0
        \end{IEEEeqnarraybox*}\right)
        \left(\begin{IEEEeqnarraybox*}[][c]{,c,}
            M_1\\
            M_2\\
            M_3\\
            M_4\\
            S_1^{A_1}\\
            S_1^{A_2}
        \end{IEEEeqnarraybox*}\right)=
        \left(\begin{IEEEeqnarraybox*}[][c]{,c,}
            S_1^{A_1}\\
            S_1^{A_2}\\
            M_2\\
            M_3
        \end{IEEEeqnarraybox*}\right)
    \end{equation*}
    Since $\mathcal{A}_1$ is full rank we have that $\mathrm{H}(S_U|M_1)=4\log_2{q}$ by Lemmas \ref{lem: 2.3} and \ref{lem: 2.4}. 

    Thus, \begin{align*}
        \mathrm{I}(M_1;S_U)&= \mathrm{H}(S_U)-\mathrm{H}(S_U|M_1)\\
        &= 4\log_2{q}-4\log_2{q}\\
        &=0.
    \end{align*}
    Showing $\mathrm{I}(M_\ell;S_U)=0$ for $\ell\in\{2,3,4\}$ involves similar arguments.
\end{proof}

Since we now have $4$ messages and we have not increased the number of shares, the information rate using our construction for this example is $\mathcal{R}_{MS}=\frac{4}{8}=\frac{1}{2}$, an improvement over the information rate $\mathcal{R}_{SS}=\frac{1}{8}$ of the single-secret sharing scheme.

\section{Main Results}

In this section, we show how to convert any single-secret sharing scheme with a general access structure into a multi-secret sharing scheme where every message has the same access structure. Given a single-secret sharing scheme with a monotone circuit, $(\mathcal{P},\Gamma_0, S)$, we define a replaceable share as the following.

\begin{definition}\label{def: 5}
    A share $S_j^A$ is \emph{replaceable} if, for every authorized subset $A'\in\Gamma_0$, either the participant $P_j\in A'$ or the subset $A\setminus P_j\subseteq A'$.
\end{definition}

Then, as shown in Section~\ref{sec: example}, our method for replacing shares consists in the following algorithm.

\begin{algorithm}[Replacement Algorithm]\label{alg: rep} Given a monotone circuit construction, we perform the following steps.
    \begin{enumerate}
        \item Identify replaceable shares $S_j^{A}$ according to Definition~\ref{def: 5}.
        \item For each authorized subset $A\in\Gamma_0$, if the share $S_j^A$ is replaceable for every participant $P_j\in A$, pick any such share to be fixed. Otherwise, fix a non-replaceable $S_j^A$.
        \item Replace every share $S_j^A$ that is replaceable but not fixed with a linear combination of messages of the form $aM_1+bM_\ell$ where $a\notin{0,1}$ and $b\neq0$ for $\ell\in\{2,\ldots,m\}$. Each $M_\ell$ is allowed to appear in exactly one replacement.
    \end{enumerate}
\end{algorithm}

We now present in Theorems~\ref{teo: construction} and~\ref{teo: nonRep} the main results of this paper for an individually secure multi-secret sharing scheme.

\begin{theorem} \label{teo: construction}
    Let $(\mathcal{P}, \Gamma_0, S)$ be a single-secret sharing scheme with a monotone circuit. Then, by applying Algorithm~\ref{alg: rep}, we construct an individually secure multi-secret sharing scheme that achieves an information rate of $\mathcal{R}_{MS}=\frac{m}{|S|}$ where $m-1$ is the number of replaced shares.
\end{theorem}

The proof of Theorem~\ref{teo: construction} relies on the following two Lemmas~\ref{lem: mainDecode} and~\ref{lem: mainSecure} to show that the new scheme is decodable and individually secure. %The proofs for these Lemmas can be found in the Appendix.

\begin{lemma}[Decodability]\label{lem: mainDecode}
    Let $S_A$ be the set of shares held by an authorized subset $A$ after applying Algorithm~\ref{alg: rep}. For all $A\in\Gamma_0$, $\mathrm{H}(M_\ell|S_A)=0$ for any $\ell$. In other words, all authorized subsets can compute all messages.
\end{lemma}

\begin{proof}
    We prove Lemma~\ref{lem: mainDecode} by showing that every message can be written as a linear combination of the shares held by an authorized subset.  By construction, $M_1$ is equal to the sum of the shares $S_j^A\in S_A$, i.e., $M_1=\sum_{P_j\in A}{S_j^A}$. Thus, $\mathrm{H}(M_1|S_A)=0$ for all $A$ by Lemma \ref{lem: 2.1}.
    
    For the remaining messages $M_\ell$, $\ell\in\{2,\ldots,m\}$, there exists a participant $P_j\in A_i$ such that $S_j^{A_i}=aM_1+bM_\ell$. For each authorized subset $A$, there are two cases to consider.
    
    \textbf{Case 1 ($P_j\in A$):} In this case, $S_j^{A_i}\in S_A$. Then, since we've already shown $A$ can compute $M_1$, we can write $M_\ell$ as a linear combination of $S_j^{A_i}$ and $M_1$. More explicitly, \\$b^{-1}(S_j^{A_i}-aM_1)=b^{-1}(aM_1+bM_\ell-aM_1)=M_\ell.$ Thus, $\mathrm{H}(M_\ell|S_A)=0$ by Lemma \ref{lem: 2.1}.
    
    \textbf{Case 2 ($P_j\notin A$):} Since $P_j\notin A$, then $A_i\setminus \{P_j\}\subseteq A$ by our replaceability condition. Let $S'=S_{A_i}\setminus S_{P_j}$ be the set of shares held by $A_i\setminus \{P_j\}$. Then $S'\subseteq S_A$. By construction, we can recover $S_j^{A_i}$ as a linear combination of the shares $S_t^{A_i}\in S'$ for $t\neq j$, i.e., $S_j^{A_i}=M_1-\sum_{t\neq j}{c_tS_t^{A_i}}$ for some constants $c_t\in\mathbb{F}_q$. Then, as in Case 1, we can write $M_\ell$ as a linear combination of $S_j^{A_i}$ and $M_1$.  Thus, $\mathrm{H}(M_\ell|S_A)=0$ by Lemma \ref{lem: 2.1}.
\end{proof}

\begin{lemma}[Security]\label{lem: mainSecure}
    After applying Algorithm~\ref{alg: rep}, no unauthorized subset gains any new information about any individual message, i.e. $\mathrm{I}(M_\ell;S_U)=0$ for $\ell\in\{1,\ldots,m\}$ where $S_U=\bigcup_{P_j\in U}{S_{P_j}}$ is the set of shares held by an unauthorized subset.
\end{lemma}
{\em Proof:} The proof is given in Appendix~\ref{proof:lem_security}.

The proof of individual security given in Appendix~\ref{proof:lem_security} is done using an inductive argument where the base case follows from the underlying secret sharing scheme. The inductive step is proved by showing that the representative matrices for $S_U$ and $S_U|M_\ell$ for $\ell\in\{1,\ldots,m-1\}$ remain full rank after the replacement since they are obtained by elementary column operations. To finish the inductive step, we show that the new representative matrix for $S_U|M_m$ is also full rank. This is done by way of contradiction. Thus, when $m-1$ shares are replaced, the information rate for our construction of a multi-secret sharing scheme is $\mathcal{R}_{MS}=\frac{m}{|S|}$.

Finally, we show that if a share $S_j^A$ not satisfying Definition \ref{def: 5} is replaced, then the decodability of the scheme fails. Thus, our replacement procedure can only be applied to replaceable shares.

\begin{theorem}\label{teo: nonRep}
    Suppose $P_j$'s share $S_j^{A_i}$ is not replaceable. If we make the replacement $S_j^{A_i}=aM_1+bM_\ell$ where $a\notin{0,1}$ and $b\neq0$, then $\mathrm{H}(M_\ell|S_A)\neq0$ where $S_A$ is the set of shares held by the participants in $A$. In other words, replacing a share that does not satisfy Definiton~\ref{def: 5} results in some authorized subset being unable to compute $M_\ell$.
\end{theorem}

{\em Proof:} The proof is given in Appendix~\ref{proof:teo_nonRep}.

The proof of Theorem~\ref{teo: nonRep} given in Appendix~\ref{proof:teo_nonRep} follows from the fact that there exists an authorized subset that does not have $S_j^A$ and also does not have enough shares of $S_A\setminus S_{P_j}$ in order to reconstruct $S_j^A$.

%\newpage

%%%%%%
%% To balance the columns at the last page of the paper use this
%% command:
%%
%\enlargethispage{-1.2cm} 
%%
%% If the balancing should occur in the middle of the references, use
%% the following trigger:
%%
\IEEEtriggeratref{25}
%%
%% which triggers a \newpage (i.e., new column) just before the given
%% reference number. Note that you need to adapt this if you modify
%% the paper.  The "triggered" command can be changed if desired:
%%
%\IEEEtriggercmd{\enlargethispage{-20cm}}
%%
%%%%%%

%%%%%%
%% References:
%% We recommend the usage of BibTeX:
%%
\bibliographystyle{IEEEtran}
\bibliography{citations}
%%
%% where we here have assumed the existence of the files
%% definitions.bib and bibliofile.bib.
%% BibTeX documentation can be obtained at:
%% http://www.ctan.org/tex-archive/biblio/bibtex/contrib/doc/
%%%%%%

\clearpage

\appendix

\off{\subsection{Proof of Lemma~\ref{lem: mainDecode}}
\begin{proof}
    By construction, $M_1$ is equal to the sum of the shares $S_j^A\in S_A$, i.e., $M_1=\sum_{P_j\in A}{S_j^A}$. Thus, $\mathrm{H}(M_1|S_A)=0$ for all $A$ by Lemma \ref{lem: 2.1}.
    For the remaining messages $M_\ell$, $\ell\in\{2,\ldots,m\}$, there exists a participant $P_j\in A_i$ such that $S_j^{A_i}=aM_1+bM_\ell$. For each authorized subset $A$, there are two cases to consider.
    \textbf{Case 1 ($P_j\in A$):} In this case, $S_j^{A_i}\in S_A$. Then, since we've already shown $A$ can compute $M_1$, we can write $M_\ell$ as a linear combination of $S_j^{A_i}$ and $M_1$. More explicitly,\[b^{-1}(S_j^{A_i}-aM_1)=b^{-1}(aM_1+bM_\ell-aM_1)=M_\ell.\] Thus, $\mathrm{H}(M_\ell|S_A)=0$ by Lemma \ref{lem: 2.1}.
    \textbf{Case 2 ($P_j\notin A$):} Since $P_j\notin A$, then $A_i\setminus \{P_j\}\subseteq A$ by our replaceability condition. Let $S'=S_{A_i}\setminus S_{P_j}$ be the set of shares held by $A_i\setminus \{P_j\}$. Then $S'\subseteq S_A$. By construction, we can recover $S_j^{A_i}$ as a linear combination of the shares $S_t^{A_i}\in S'$ for $t\neq j$, i.e., $S_j^{A_i}=M_1-\sum_{t\neq j}{c_tS_t^{A_i}}$ for some constants $c_t\in\mathbb{F}_q$. Then, as in Case 1, we can write $M_\ell$ as a linear combination of $S_j^{A_i}$ and $M_1$.  Thus, $\mathrm{H}(M_\ell|S_A)=0$ by Lemma \ref{lem: 2.1}.
\end{proof}}

\subsection{Proof of Lemma~\ref{lem: mainSecure} - Individual Security}\label{proof:lem_security}

\begin{proof}
    Let $R=\{S_1,\ldots,S_w\}$ be the random shares that are distributed to the participants. Also, let $S_{u_1},\ldots,S_{u_t}$ denote the shares in $S_U$. Note that, since $S_U\subseteq S$, we have $\lvert S_U\rvert=t < w+1=\lvert S\rvert+1$.
    
    Without making any replacements, $\mathrm{I}(M_1;S_U)=\mathrm{H}(S_U)-\mathrm{H}(S_U|M_1)=0$ from the underlying secret sharing scheme \cite{stinson}. Then $\exists$ a representative matrix $\mathcal{A}\in\mathbb{F}_q^{t\times w+1}$ of $S_U$ such that 
    \begin{equation*}
        \mathcal{A}
        \left(\begin{IEEEeqnarraybox*}[][c]{,c,}
            M_1\\
            S_1\\
            S_2\\
            \vdots\\
            S_w
        \end{IEEEeqnarraybox*}\right)=
        \left(\begin{IEEEeqnarraybox*}[][c]{,c,}
            S_{u_1}\\
            S_{u_2}\\
            \vdots\\
            S_{u_t}
        \end{IEEEeqnarraybox*}\right)
    \end{equation*}
    Then to get the representative matrix $\mathcal{A}_1\in\mathbb{F}_q^{t\times w+1}$ of $S_U|M_1$, multiply the first column of $\mathcal{A}$ by $0$. Since $\mathrm{H}(S_U)-\mathrm{H}(S_U|M_1)\leq t\log_2{q}-\mathrm{H}(S_U|M_1)=0$, we must have that $\mathcal{A},\mathcal{A}_1$ both full rank.
    
    Now suppose that we've replaced $m-2$ shares and the representative matrices of $S_U$ and $S_U|M_\ell$ are full rank for any $\ell \in \{1,\ldots,m-1\}$. In other words, suppose $\mathrm{I}(M_\ell;S_U)=0$ for any $\ell \in \{1,\ldots,m-1\}$.
    
    Let $\mathcal{A}$ be the representative matrix of $S_U$ after replacing $m-2$ shares, i.e., 
    \begin{equation*}
        \mathcal{A}
        \left(\begin{IEEEeqnarraybox*}[][c]{,c,}
            M_1\\
            \vdots\\
            M_{m-1}\\
            S_{m-1}\\
            \vdots\\
            S_w
        \end{IEEEeqnarraybox*}\right)=
        \left(\begin{IEEEeqnarraybox*}[][c]{,c,}
            S_{u_1}\\
            S_{u_2}\\
            \vdots\\
            S_{u_t}
        \end{IEEEeqnarraybox*}\right)
    \end{equation*}
    
    When we replace $S_{m-1}$ with a linear combination of the form $aM_1+bM_m$,  $S_U$ can now be represented by
    \begin{equation*}
        \mathcal{A}'
        \left(\begin{IEEEeqnarraybox*}[][c]{,c,}
            M_1\\
            \vdots\\
            M_{m-1}\\
            M_m\\
            S_m\\
            \vdots\\
            S_w
        \end{IEEEeqnarraybox*}\right)=
        \left(\begin{IEEEeqnarraybox*}[][c]{,c,}
            S_{u_1}\\
            S_{u_2}\\
            \vdots\\
            S_{u_t}
        \end{IEEEeqnarraybox*}\right)
    \end{equation*}
    where $\mathcal{A}'$ is obtained by applying the column operations $aC_m+C_1$ and $bC_m$ to $\mathcal{A}$. Since column operations preserve rank, $\mathcal{A}'$ is still full rank. Similarly, if $\mathcal{A}_\ell$ is the representative matrix of $S_U|M_\ell$ for $\ell\in\{1,\ldots,m-1\}$, then, after replacing $S_{m-1}$, the new representative matrix of $S_U|M_\ell$ is obtained by applying the column operation $bC_m$ to $\mathcal{A_\ell}$. Call these new full rank matrices $\mathcal{A}_\ell'$. Lastly, the representative matrix $\mathcal{A}_m\in\mathbb{F}_q^{t\times w+1}$ of $S_U|M_m$ is obtained by multiplying the $m^{th}$ column of $\mathcal{A}'$ by 0. Suppose, FTSOC, that $\mathcal{A}_m$ is not full rank. Let $C_1,\ldots,C_{w+1}$ be the columns of $\mathcal{A}$ and $C_1',\ldots,C_{w+1}'$ be the columns of $\mathcal{A}_m$. Since $\mathcal{A}_m$ is not full rank, $\exists$ a linear combination $\beta_1C_1'+\cdots+\beta_mC_m'+\cdots+\beta_{w+1}C_{w+1}'=0$ such that $\beta_i\neq0$ for some $i$. Then \[\beta_1C_1'+\cdots+\beta_mC_m'+\cdots+\beta_{w+1}C_{w+1}'\] \[=\beta_1(aC_m+C_1)+\cdots+\beta_m(0(bC_m))+\cdots+\beta_{w+1}C_{w+1}\] \[=\beta_1C_1+\cdots+\beta_maC_m+\cdots+\beta_{w+1}C_{w+1}=0.\] Thus $\mathcal{A}$ is not full rank, which is a contradiction. Hence, $\mathcal{A}_m$ must be full rank. Therefore, $\mathrm{H}(S_U)=\mathrm{H}(S_U|M_\ell)=t\log_2{q}$ for $\ell\in\{1,\ldots,m\}$ by Lemmas \ref{lem: 2.3} and \ref{lem: 2.4}. Thus, $\mathrm{I}(M_\ell;S_U)=\mathrm{H}(S_U)-\mathrm{H}(S_U|M_\ell)=t\log_2{q}-t\log_2{q}=0$ for $\ell\in\{1,\ldots,m\}$.
\end{proof}

\subsection{Proof of Theorem~\ref{teo: nonRep} - Non-Replaceable}\label{proof:teo_nonRep}

\begin{proof}
    A share $S_j^{A_i}$ being non-replaceable implies that the participant $P_j\notin A$ and the subset $A_i\setminus P_j\not\subseteq A$ for some authorized subset $A\in \Gamma$.

    Since $A$ is an authorized subset, $A$ can compute $M_1$ as a linear combination of the shares in $S_A$.

    Since the participant $P_j\notin A$, then their share $S_j^{A_i}\notin S_A$. Also, since $A_i\setminus\{P_j\}\not\subseteq A$, there must exist (at least one) $P_\alpha\in A_i\setminus\{P_j\}$ such that $P_\alpha\notin A$. Then their share $S_\alpha^{A_i}\notin S_A$. Since the shares in $S_{A_i}$ sum to $M_1$, we can write $S_j^{A_i}+S_\alpha^{A_i}+\sum_{P_t\in A\cap A_i}{S_t^{A_i}}=M_1$. By replacing $S_j^{A_i}$ with $aM_1+bM_\ell$ and rearranging this equation, we are able to get $M_\ell+S_\alpha^{A_i}$ as a linear combination of $M_1$ and the shares $S_t^{A_i}$ where $t$ is such that $P_t\in A\cap A_i$. In other words, $M_\ell+S_\alpha^{A_i}=\sum_{P_t\in A\cap A_i}{c_tS_t^{A_i}}+cM_1$ for some constants $c,c_t\in \mathbb{F}_q$. Then 
    \begin{align}
        &\mathrm{H}(M_\ell|S_A)\\
        &=\mathrm{H}\left(\sum_{P_t\in A\cap A_i}{c_tS_t^{A_i}}+cM_1-S_\alpha^{A_i}|S_A\right)\\ \label{eq: 6} &=\mathrm{H}(-S_\alpha^{A_i})\\ 
        \label{eq:7} &=\log_2{q}
    \end{align} 
    where \eqref{eq: 6} follows from Lemma \ref{lem: 2.2} since $S_t^{A_i}\in S_A$ for $P_t\in A\cap A_i$ and $M_1$ is a linear combination of the shares in $S_A$. Finally, \eqref{eq:7} follows from Lemma \ref{lem: 2.3}.
\end{proof}

\subsection{Preliminaries}
Here, we provide a few relevant Lemmas for the results presented in this work.

\begin{lemma} \label{lem: 2.1}
    If $Y=f(X)$, then $\mathrm{H}(Y|X)=0$. 
\end{lemma}

\begin{proof}
    By assumption, we see that
    \begin{align} \label{eq: 3.1}
        \mathrm{H}(Y|X) &=\mathrm{H}(f(X)|X)\\ \label{eq: 3.2} &=-\sum_{x\in X}\p[X=x]\mathrm{H}(f(X)|X=x)\\ \label{eq: 3.3} &=0
    \end{align}

    Where ~\eqref{eq: 3.2} follows from the definition of conditional entropy. Finally, ~\eqref{eq: 3.3} follows from the fact that $\mathrm{H}(f(X)|X=x)=0$. To see this, note that \begin{align*}&\mathrm{H}(f(X)|X=x)\\&=-\sum_{y\in Y}{\p[f(X)=y|X=x]\log_2{(\p[f(X)=y|X=x])}}\end{align*} and $\p[f(X)=y|X=x]=\p[f(x)=y]$ which will either be $0$ or $1$ for all $y\in Y$.   
\end{proof}

In the context of this paper, Lemma \ref{lem: 2.1} implies that, if a message can be written as a linear combination of a given set of shares, then the entropy of that message given those shares is $0$. In other words, for decodability, we need only show that every message is a linear combination of shares held by an authorized subset. The remaining Lemmas are relevant for proving the individual security of a scheme.

\begin{lemma} \label{lem: 2.2}
    If $X$ and $Y$ are independent, $\mathrm{H}(X+Y|X)=\mathrm{H}(Y)$.
\end{lemma}

\begin{proof}
    We want to show that the conditional entropy of $X+Y$ given $X$ is equal to the entropy of $Y$. By definition of conditional entropy, we get
    \begin{align*} 
        &\mathrm{H}(X+Y|X)\\&=\sum_{x\in X}{\p[X=x]\mathrm{H}(X+Y|X=x)} \\
        &= \begin{multlined}[t]
            -\sum_{x\in X}{\p[X=x]}\\
            \sum_{x+y\in X+Y}\p[X+Y=x+y|X=x]\\
            \qquad\qquad\qquad\log_{2}{(\p[X+Y=x+y|X=x]})
        \end{multlined}  \\
        &= -\sum_{x\in X}{\p[X=x]}\sum_{y\in Y}{\p[Y=y]\log_{2}(\p[Y=y])} \\
        &= \sum_{x\in X}{\p[X=x]}\mathrm{H}(Y) \\
        &= \mathrm{H}(Y)
    \end{align*}
\end{proof}

\begin{lemma} \label{lem: 2.3}
    If $X$ is uniformly distributed over $\mathbb{F}_q^m$, then $\mathrm{H}(X)=m\log_2{q}$.
\end{lemma}

\begin{proof}
     Since $X$ is uniformly distributed over $\mathbb{F}_q^m$, we have that $\p[X=x]=\frac{1}{q^m}$ for all $x\in\mathbb{F}_q^m$. Then 
    \begin{align*}
        \mathrm{H}(X)&=-\sum_{x\in X}{\p[X=x]\log_2{\p[X=x]}}\\
        &=-\sum_{x\in\mathbb{F}_q^m}{\frac{1}{q^m}\log_2{\frac{1}{q^m}}}\\
        &=-q^m\frac{1}{q^m}\log_2{\frac{1}{q^m}}\\
        &=-(\log_2{1}-\log_2{q^m})\\
        &=m\log_2{q}
    \end{align*}

    When $m=1$, this is equivalent to $\mathrm{H}(X)=\log_2{q}$ if $X$ is uniformly distributed over $\mathbb{F}_q$.
\end{proof}

\begin{lemma} \label{lem: 2.4}
    Let $\mathcal{A}\in\mathbb{F}_q^{m\times n}$ with $m<n$ such that $\mathcal{A}$ has full rank, $X\in\mathbb{F}_q^n$ with every entry uniformly distributed over $\mathbb{F}_q$, and let $\mathcal{A}X=Y\in\mathbb{F}_q^m$. Then $Y$ is uniformly distributed over $\mathbb{F}_q^m$.
\end{lemma}

\begin{proof}
    Consider $\mathcal{A}'=\left(\begin{smallmatrix}
        \mathcal{A}\\ B_{m+1}\\ \vdots\\ B_n
    \end{smallmatrix}\right)\in\mathbb{F}_q^{n\times n}$ where the rows $B_{m+1},\ldots,B_n$ are chosen to be linearly independent and linearly independent from the rows of $\mathcal{A}$. Then $\mathcal{A}'$ is invertible. Let $\mathcal{A}'X=Y'\in\mathbb{F}_q^n$. Since each element of $X$ is uniformly distributed over $\mathbb{F}_q$, every element of $Y'$ is also uniformly distributed $\mathbb{F}_q$. To see this, $\p[Y'_i=y'_i]=\p[X_i=(\mathcal{A}'^{-1}y')_i]=\frac{1}{q}$ for all $1\leq i\leq n$. Note that $Y'=\left(\begin{smallmatrix}
        Y\\ B_{m+1}X\\ \vdots\\ B_nX
    \end{smallmatrix}\right)$. Let $Y'_{[m]}$ denote the first $m$ entries of $Y'$. Then $\p[Y=y]=\p[Y'_{[m]}=y]=\p[Y'_1=y_1]\cdots\p[Y'_m=y_m]=\frac{1}{q^m}$. Therefore, $Y$ is uniformly distributed over $\mathbb{F}_q^m$.
\end{proof}

\end{document}